%% file: main.tex
\documentclass[
    headings=standardclasses
    ]{scrartcl}

\input{setup.tex}
\usepackage{cite}

\ACMfalse

\input{common.tex}

\begin{document}

\maketitle

\begin{abstract}
    \abstractContents
\end{abstract}

\input{contents.tex}

\end{document}

%% file: setup.tex
\usepackage{amsfonts}
\usepackage{amsmath}
\usepackage{amsthm}
\usepackage{bbold}
\usepackage{multido}
\usepackage[normalem]{ulem}
\usepackage{url}
\usepackage{graphicx}

\usepackage{listings}
\newif\ifACM

%% file: common.tex
\ifACM
\else
    \usepackage[hidelinks]{hyperref}
\fi
\usepackage{cleveref} 

\title{SACR\'{E} BLEU: Self-Assessed Creator Royalties \'{E}nforced by Balancing Liquidity Estimation \& Utility}

\subtitle{A formal definition and analysis of Ethereum Request for Comment (ERC) 7526}

\newcommand{\abstractContents}{
    The secondary market for Ethereum non-fungible tokens (NFTs) has resulted in over \$1.8bn being paid to creators in the form of a sales tax commonly called creator royalties.
    This was despite royalty payments being enforced by no more than social contract alone.
    Predictably, such an incentive structure led to zero-royalty alternatives becoming abundant and payments dwindled.
    A purely programmatic solution to royalty enforcement is hampered by the prevailing NFT standard, ERC\nobreakdash-721, which is ignorant of sale values and royalty enforcement therefore relies on (potentially dishonest) third parties.
    We thus introduce an incentive-compatible mechanism for which there is a single rationalisable solution, in which royalties are paid in full, while maintaining full ERC\nobreakdash-721 compatibility.
    The mechanism constitutes the core of ERC\nobreakdash-7526.
}

\newcommand{\emailaddr}[1]{#1@solidifylabs.com}
\newcommand{\emaillink}[1]{\href{mailto:\emailaddr{#1}}{\emailaddr{#1}}}
\newcommand{\equalContrib}{Both authors contributed equally.}

\ifACM
    \acmSubmissionID{289}

    \author{David Huber}
    \email{\emailaddr{dave}}
    \affiliation{
        \institution{Solidify Labs}
        \country{Austria}
    }
    \authornote{\equalContrib}

    \author{Arran Schlosberg}
    \email{\emailaddr{arran}}
    \authornotemark[1]
\else
    \newcounter{savecntr}
    \newcounter{restorecntr}
    \author{
        David Huber\footnote{\equalContrib}\\
        \texttt{\emaillink{dave}}
        \and
        Arran Schlosberg\setcounter{restorecntr}{\value{footnote}}\setcounter{footnote}{\value{savecntr}}\footnotemark\\
        \texttt{\emaillink{arran}}
        \setcounter{footnote}{\value{restorecntr}}
    }
\fi

\theoremstyle{plain}
\newtheorem{lem}{Lemma}
\newtheorem{cor}{Corollary} 
\newtheorem{thm}{Theorem}

\theoremstyle{remark}
\newtheorem{rem}{Remark}

\theoremstyle{definition}
\newtheorem{rul}{Rule}
\newtheorem{des}{Requirement}
\newtheorem{dfn}{Definition}

\crefname{cor}{corollary}{corollaries}
\crefname{des}{requirement}{requirements}
\crefname{rul}{rule}{rules}
\crefname{dfn}{definition}{definitions}

%% file: contents.tex
\newcommand\numberthis{\addtocounter{equation}{1}\tag{\theequation}}

\newcommand{\erc}[1]{ERC\nobreakdash-{#1}}

\newcommand{\player}[1]{\mathsf{#1}} 
\newcommand{\curr}{\player{P}}
\newcommand{\prev}{\player{O}}
\renewcommand{\P}{\curr}
\newcommand{\Q}{\prev}
\renewcommand{\O}{\prev}

\newcommand{\paid}{a}
\newcommand{\qeq}{\overset{?}{=}}
\newcommand{\val}{v}
\newcommand{\fmv}{m}
\newcommand{\fee}{\phi}
\newcommand{\price}{\pi}
\newcommand{\priceInv}{\price^{-1}}
\newcommand{\cost}{c}
\newcommand{\rev}{x}
\newcommand{\abs}[1]{\lvert#1\rvert}
\newcommand{\coll}{\kappa}
\newcommand{\bribe}{\beta}
\newcommand{\roi}{\lambda}

\section{Introduction}

\subsection{Background}

How do you charge a royalty on the sale of an item when you don't know the price paid?
How do you enforce a sales tax when there's no proof of a sale having even occurred?

The ability to receive a sales tax from secondary sales of non-fungible tokens (NFTs), commonly called creator royalties, promised creators a chance to benefit from their work accruing in value.
Without royalties, an artist selling their piece for some fixed initial price is left watching others profit as the value of their creation grows over time.
The exclusion of creators from the fruits of their labour, even those only realised later, is a disincentive to their market participation.

As of October 2022, \$1.8bn in NFT royalties had been paid\cite{GalaxyNFTRoyalties}.
Despite many creators participating in the false belief that payments were programmatically enforced, they were in fact a social contract upheld by goodwill and mere ignorance of alternatives.
Unsurprisingly, a race to zero royalties ensued\cite{NansenNFTRoyalties}.
Although this suggests that enforcement needs to be handled by smart contracts, nuances in the definition of the Ethereum NFT standard, \erc{721}\cite{ERC721}, make this difficult.

\erc{721} doesn't explicitly cater for token sales.
Instead, it allows an NFT owner to approve a third-party \emph{operator} to transfer the token on their behalf.
By approving marketplace contracts with verifiable rules around trades, owners can list tokens for sale---tokens are transferred if and only if the requisite consideration changes hands.
This is typically the point at which the marketplace also forwards a portion of the sale price as a royalty.
Note, however, that \erc{721} contracts are ignorant of the transfer of funds in this scenario and there is no computational requirement to make them aware.
From the perspective of an \erc{721}-compliant contract, there is no difference between a token sale, gift, or transfer to an address controlled by the same entity.
As not all marketplaces respect payment of sales taxes, creators must therefore trust or force buyers and sellers---the ultimate payers of the royalties---to trade on marketplaces that act in the best interests of the creators.

We thus introduce an incentive-compatible mechanism to enable NFT sales taxes with decentralised enforcement.
The mechanism constitutes the core of \erc{7526}\cite{ERC7526} and, unlike other proposed solutions, can be implemented without altering \erc{721}'s approved transfers.

\subsection{Related Work}

We are unaware of any proposed methods of royalty enforcement that act to inventivise participants through mechanism design.
All commonly used solutions are either limited to communicating desired royalty amounts, or somehow disable \erc{721} functionality to provide programmatic enforcement.

\subsubsection{ERC-2981}

Perhaps the most widely adopted of the royalty solutions, \erc{2981}\cite{ERC2981} is an extension to \erc{721} that allows NFT contracts to signal expected royalty payments.
It therefore only acts to broadcast information but not to enforce payment, and the authors stated that ``payment must be voluntary, as transfer mechanisms... [do] not always imply a sale occurred''---this observation is key to the difficulty of programmatic enforcement and was a primary motivator of our work.
While their contribution allowed NFT creators to control the source of truth for calculation of expected royalties, it still relied on third parties (with different incentives) to honour the suggested payment.

\subsubsection{Royalty Registry}

Although the introduction of \erc{2981} empowered NFT creators, its functionality couldn't be applied retroactively to existing contracts as the majority are immutable and non-upgradable.
To counter this, the Royalty Registry\cite{RoyaltyRegistry} acts as a single source of truth, accessible on-chain, for all NFT contracts.
It honours \erc{2981}-compliant contracts, acting merely as a proxy in such cases, otherwise allowing the owners of older NFT contracts to register royalty parameters.
As another broadcast-only approach, it shares the same shortcomings as \erc{2981}.

\subsubsection{Operator filters}

An easily implemented solution to combatting dishonest actors is simply to block them from being allowed to act as approved operators\cite{OperatorFilter,ERC721C}.
This, however, was broadly abandoned\cite{OpenSeaOptionalRoyalties}.
Beyond possible circumvention of the process, breaking technical guarantees of permissionless integrations---disallowing holders from selling their assets on any platform---was considered antithetical to the ethos of the ecosystem.

\subsubsection{ERC-4910}

Less-broadly known than \erc{2981}, a second royalty standard was defined in \erc{4910}\cite{ERC4910}.
An inspection of its reference implementation\cite{ERC4910ReferenceImpl} reveals that it achieves royalty enforcement by disabling some transfer mechanisms, which breaks fundamental \erc{721} compatibility, leaving users unable to trade their tokens on standard marketplaces.

\section{Problem}

\subsection{Definitions}

\begin{dfn}[Pseudonyms]
    Each player $\P$ has a set of one or more blockchain addresses $A_\P$ that act as pseudonyms through which the player acts---we say that the player \emph{controls} each address ${a \in A_\P}$.
    The player controlling a specific address is private and known only to $\P$.
    We assume that players control disjoint sets of addresses; i.e. ${\O \neq \P \iff A_\O \cap A_\P = \emptyset}$.
\end{dfn}

\begin{dfn}[Ownership]
    A token $t$ is owned by a single address $a_t$, which is known to everyone.
    There is therefore a transitive relationship from token to player: ${a_t \in A_{\P_t}}$. The player $\P_t$ can also be said to \emph{own} the token but this is similarly private, known only to $\P$.
    When obvious from context that we are referring to a token owner, the $_t$ subscript is elided for brevity.
\end{dfn}

\begin{dfn}[Transfer]
    The transfer of a token occurs when there is a change in ownership from $a'_t$ to $a_t$ where ${a'_t \neq a_t}$.
    This may or may not result in a change of player.
    Similarly, a string of consecutive transfers results in owners ${a^{0}_t,\dots,a^n_t}$ with a change of player at each step being possible but not necessary.

    To aid in understanding, whenever a transfer occurs, we denote the old owning address as $a'_t$ and the new one as $a_t$.
    Similarly, the old owning player will be $\O$ and the new one $\P$, regardless of whether $\O \qeq \P$ is known, as inferring inequality thereof is a motivating factor of this work.
\end{dfn}

\begin{dfn}[Cost]
    The transfer of a token incurs a cost $\cost$ for the receiving player $\P$.
    This cost might be zero, e.g. for gifts or transfers between addresses controlled by the same player.
\end{dfn}

\begin{dfn}[Revelation]
    \label{dfn:reveal}
    Upon token transfer the, receiving player $\P$ will reveal information about the nature of the transfer by either disclosing a value $\rev \in \mathbb{R}$ to the mechanism, or by electing not to do so before their turn expires.
\end{dfn}

\begin{dfn}[Fee function]
    \label{dfn:fees}
    There is a token-specific fee function $0 < \fee_t(\rev)$ that is strictly monotonically increasing in $\rev$ and has a best Lipschitz constant $L_\fee < 1$ (i.e. the fee\footnote{That is, fee $\fee$ (fo fum).} grows slower than the input parameter).
    This fee is always paid to the mechanism upon the disclosure described in \Cref{dfn:reveal}.
\end{dfn}

\begin{dfn}[Historical non-fee-paying owners]
    Consider the sequence of token-owning addresses beginning with the last one to pay a fee, or from token creation if no fee has ever been paid.
    Let $H$ be the set of distinct addresses in this historical sequence, excluding the current owner:
    \begin{equation*}
        H = \{\paid^{k},a^{k+1},\dots,a^{n-1}\}\text{, where }n,k\in \mathbb{N}_0 \text{, }k < n\text{, and }a^n = a_t
        .
    \end{equation*}
    If the current owner has disclosed some $\rev$ and paid $\fee({\rev})$, then $H = \emptyset$.
\end{dfn}

\begin{dfn}[Auto-sale price]
    \label{dfn:price}
    There is a token-specific price function $0 < \pi_t(\rev)$ that is strictly monotonically increasing in $\rev$.
\end{dfn}

\begin{dfn}[Valuation]
    \label{dfn:valuation}
    Each player $\P$ has two valuations of the token, which may differ due to information assymetry:
    \begin{itemize}
        \item $\fmv_\P$ denotes their estimate of the token's free-market value, i.e. the price below and above which a buyer will or will not be found right now, respectively; and
        \item $\val_\P$ denotes their self-assessed valuation, in current terms, of hodling the token long-term.
    \end{itemize}
\end{dfn}

\subsection{Desiderata}
\label{sec:desiderata}

\begin{des}[Truthful revelation]
    %
    %
    %
    \label{des:revelation}
    For every transfer of the token $t$ that results in a change of the owning player, i.e. ${\P \neq \Q}$, the best-response strategy of the new owner $\P$ must result in:
    \begin{enumerate}
        \item Disclosure of some value $\rev$ such that $\price(\rev) = \fmv_\P$ (i.e. $\P$ truthfully revealing the change in player, coupled with their true belief of the free-market value of $t$); and
        \item $\P$ paying a royalty of $\fee(\rev)$ to the mechanism.
    \end{enumerate}
\end{des}

\begin{des}[Self-transfer]
    \label{des:self-transfer}
    Should a player transfer a token between addresses that they control, i.e. both ${a, a' \in A_\P}$, then no such disclosure nor fee are required.
\end{des}

\begin{rem}
    \label{rem:self-transfer}
    A player may choose to have a token owned by any address under their control. From the perspective of the mechanism, the rationale is irrelevant, suffice to require that no fee is required should the ``beneficial owner'' (the player) remain unchanged. Similarly, we assume that there is no change in utility for a player moving a token between their own addresses\footnote{In real-world situations this is common practice for a variety of non-economic reasons, such as vanity, and must therefore be allowed by the mechanism.}.
\end{rem}


\section{Mechanism}

\subsection{Rules}

\begin{rul}[First-move]
    \label{rul:first}
    After token transfer, the owner $\P$ moves first, electing whether or not to disclose some $x$, pay fees according to $\phi(x)$ and collapse $H$.
\end{rul}

\begin{rul}[Take-back]
    \label{rul:takeback}
    Any address ${a \in H}$ may assert ownership of the token and transfer it from $a_t$ to itself, free of charge and without refunding $\cost$, restarting the mechanism at \Cref{rul:first} with the taker assuming the role of $\P$.
\end{rul}

\begin{rul}[Auto-sale]
    \label{rul:autosale}
    This rule only applies if the current owner $a_t$ chose to disclose $\rev$ and pay $\fee(\rev)$.
    For some set period of time, any address---regardless of membership in $H$---may purchase the token from $a_t$ for price $\price(\rev)$.
    Such a sale restarts the mechanism at \Cref{rul:first} with the new owner\footnote{If the auto-sale purchaser were not liable for a fee like any other address then a malicious player would disclose $\priceInv(0)$ and immediately pay $0$ to bypass the mechanism.}, and also reimburses the fee paid by $a_t$ in \Cref{rul:first}.
    Since the option to purchase expires\footnote{Auto-sale is reminiscent of a COST / Harberger tax, however, it is important to note that it expires. The mechanism is better understood as \emph{eventual} ownership than as \emph{partial}.}, we model it as a discrete turn.
\end{rul}

\subsection{Analysis}

We prove that the proposed mechanism induces all players to act truthfully under a best-response strategy.
This results in an unambiguous sequence of moves, with a single rationalisable solution, under which a royalty fee is paid iff the player owning the token changes.
This royalty fee is a function of the owning player's truthfully disclosed estimation of fair-market value.
Unless stated otherwise we assume that side-band collusion between players is impossible, but do demonstrate bounds on collusion and general royalty avoidance.

\begin{lem}
    \label{lem:takeback}
    Invoking \Cref{rul:takeback} to reclaim ownership of a token is the best response for player $\Q$ if player $\P \neq \Q$ chooses not to pay any fee in \Cref{rul:first}.
\end{lem}

\begin{proof}
    Player $\Q$ gains $\max \left(\fmv_\Q, \val_\Q \right) \ge 0$ for no cost by invoking the takeback, whereas not doing so has no benefit.
    A free lunch---sacr\'{e} bleu!
\end{proof}

\begin{cor}
    \label{cor:reveal}
    If a transfer results in a \emph{change} of controlling player to $\P \neq \Q$, their best response is to reveal some value $\rev$ to protect against takeback.
\end{cor}

\begin{lem}
    \label{lem:hplayer}
    In equilibrium, every address in $H$ is controlled by the same player. This player is the token owner; i.e. $H \subset A_\P$.
\end{lem}

\begin{proof}
    By contradiction. If $A_\O \cap H \neq \emptyset$ for any $\O \neq \P$, player $\O$ will invoke the take-back as proven in \Cref{lem:takeback} and collapse $H$ by paying a fee less than the utility gained from the takeback.
\end{proof}

\begin{rem}
    When control of a token is transferred to a different player, although $\P$ and $\Q$ lack explicit knowledge of each other's identities, both know that they are different players because of the disjoint nature of their address sets. Notably, \emph{only} these players know that $\P \neq \Q$ and the mechanism can therefore only induce these parties to reveal this information.
\end{rem}

\begin{cor}
    \label{cor:payiff}
    Under the best-response strategy of the receiving player $\P$, revealing a value $x$ is a proxy for a change in owning player; i.e. $H \subset A_{\Q} \xrightarrow{reveal} H \subset A_{\P} \iff \P \neq \Q$.
\end{cor}

\begin{lem}
    \label{lem:truevalue}
    Under \Cref{rul:autosale}, for auto-sale price $\price(\rev)$, the best response of player $\P$ is to truthfully disclose their fair-market valuation via $\rev_\P = \priceInv(\fmv_\P)$.
\end{lem}

\begin{proof}
    Player $\P$ will disclose $x$ if they intend to either keep the token long-term, or sell it again in the auto-sale window to exploit an arbitrage between different estimations of the free market value.
    Recalling that an auto-sale triggers a reimbursement of the fee, it results in quasilinear utility
    \begin{align*}
        u_\text{resell}(\rev) := \price(\rev) - c
        .
        \numberthis
    \end{align*}
    Conversely, keeping the token for some time exploits the token's future value at the cost of paying a fee now, yielding utility
    \begin{align*}
        u_\text{keep}(\rev) := \val_\P - \fee(\rev) - c.
        \numberthis
    \end{align*}
    The aggregate utility expected by player $\P$ can thus be expressed as a piecewise function dependent on their belief of the token's fair-market value
    \begin{align*}
        U_\P(\rev) :=
        \begin{cases}
            u_\text{resell}(\rev) & \pi(x) < \fmv_\P \\
            u_\text{keep}(\rev)   & \pi(x) > \fmv_\P
        \end{cases}
        \numberthis
        .
    \end{align*}
    Using the strict monotonicity of $\fee$ and $\price$ as given in \Cref{dfn:fees,dfn:price}, one finds the supremum of the utility at the crossover point $\rev_\P := \priceInv(\fmv_\P)$ between the two branches of the function as
    \begin{align*}
        U_\P^* := \sup_\rev  U_\P(\rev)
         & = \max \left( u_\text{resell}(\rev_\P), u_\text{keep}(\rev_\P) \right) \\
         & = \max \left( \price(\rev_\P),  \val_\P - \fee(\rev_\P) \right) - c    \\
         & = \max \left( \fmv_\P,  \val_\P - \fee(\rev_\P) \right) - c
        \numberthis
        .
        \label{eq: supremum utility}
    \end{align*}
    In both situations the player obtains maximum utility in the limit $\rev \rightarrow \rev_\P$. In the resale scenario, $\rev$ approaches from below and maximises sale revenue, whereas in the token-keeping scenario it approaches from above and minimises the fee.
\end{proof}

\begin{rem}
    For an alternate framing of the proof of \Cref{lem:truevalue}, compare the utility of misreporting $\rev' \neq \rev_\P$ against $U_\P^*$. By definition, underreporting $\rev' < \rev_\P$ results in the player believing that a resale will occur, and the player obtains
    \begin{align*}
        U_\P(\rev' < \rev_\P)  = \price(\rev') - c
        < \price(\rev_\P) - c
        \leq U_\P^*
        \numberthis
        \label{eq: utility underreporting}
        .
    \end{align*}
    Conversely, overreporting $\rev' > \rev_\P$ demonstrates that the player wishes to keep the token, yielding
    \begin{align*}
        U_\P(\rev' > \rev_\P)  = \val_\P - \fee(\rev') - c
        <  \val_\P - \fee(\rev_\P) - c
        \leq U_\P^*
        \numberthis
        \label{eq: utility overreporting}
        .
    \end{align*}
\end{rem}

\begin{rem}
    The result in \Cref{lem:truevalue} demonstrates why the move away from royalty payments occurred in the first place.
    Participation when wishing to keep the token results in $v_\P - \fee(\rev_\P)$ instead of $v_\P$, providing an incentive to avoid $\fee$.
\end{rem}

\begin{thm}[Best-response incentive-compatibility]
    \label{thm:incentive-compat}
    Upon transfer of a token from an address controlled by player $\Q$ to one controlled by player $\P$, there is a single rationalisable solution under best responses, such that $\P$ truthfully reveals whether or not ${\P \qeq \Q}$ and, if they differ, truthfully reveals their valuation $\fmv_\P$ by dislosing $x$ in the close neighborhood of $\priceInv(\fmv_\P)$ (see also \Cref{fig:game}).
\end{thm}

\begin{proof}
    \Cref{cor:payiff,cor:reveal} of \Cref{lem:hplayer} prove the truthful revelation of ${\P \qeq \Q}$ and, if the players differ, the disclosure of \emph{some} value $\rev$.
    \Cref{lem:truevalue} proves that the player obtains maximum utility in the limit $\rev \rightarrow \priceInv(\fmv_\P)$, making disclosure of $\rev$ in the close neighborhood of $\priceInv(\fmv_\P)$ the best response.
\end{proof}

\begin{figure}
    \includegraphics[width=\textwidth, trim={0 3.4in 0 0}, clip]{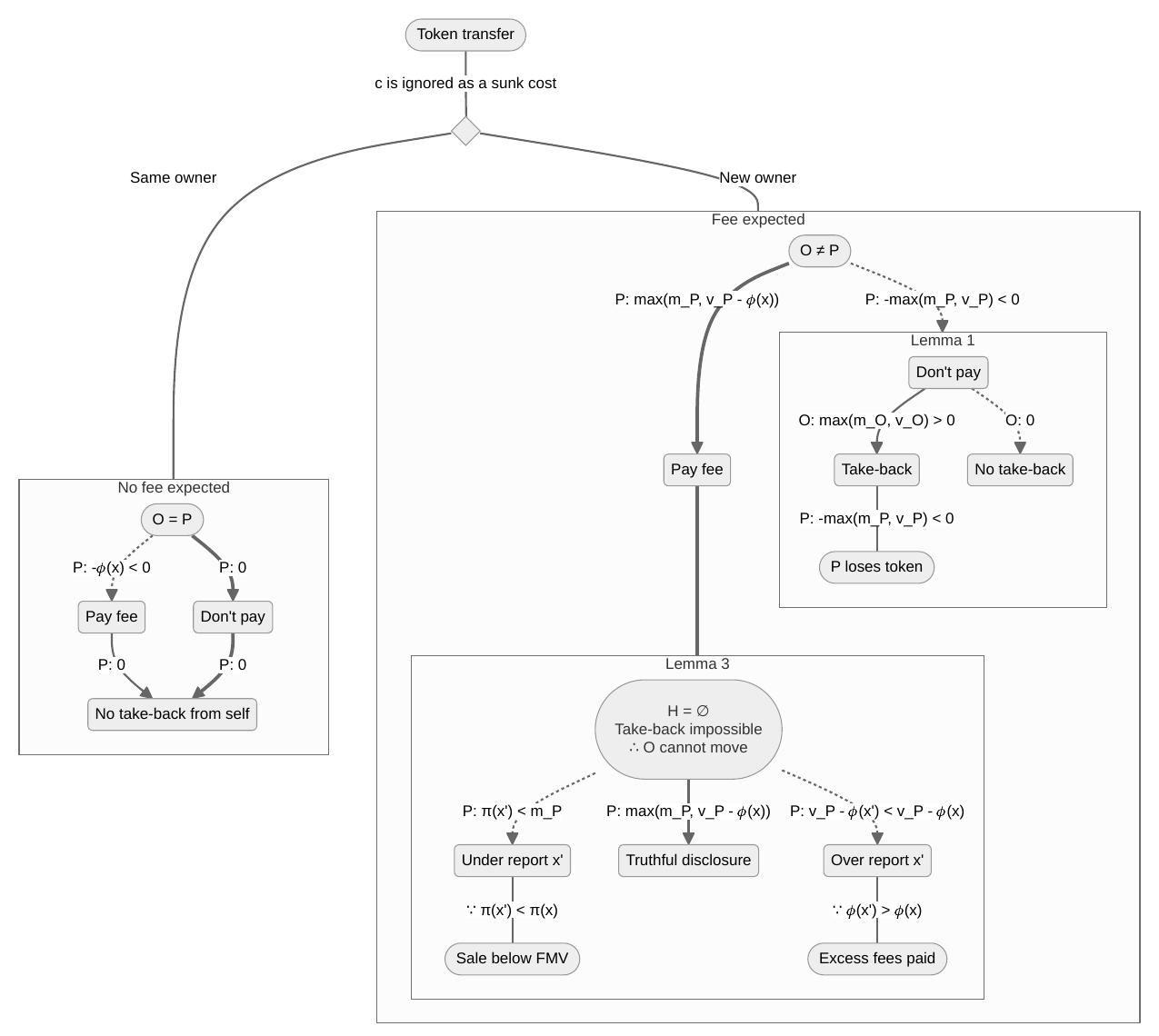}
    \caption{
        \label{fig:game}
        Flowchart representation of \Cref{thm:incentive-compat}.
        Under best responses, after a token is transferred between addresses, the player controlling the receiving address pays a royalty fee iff it wasn't a self-transfer.
        The royalty is computed on the player's truthfully revealed estimate of fair-market value. \\
        Lines with arrows depict hypothetical moves while those without them constitute passive flow due to branching scenarios or automatic consequences. Solid lines represent best responses from their originating state and bold ones follow the path of best responses to a rationalisable solution. \\
        Moves are annotated with the player ($\O$ or $\P$) and their best possible outcome, having commited to the move, but from then on following a best-response strategy.
    }
\end{figure}

\begin{dfn}[Collusion]
    A collusion is a contractual arrangement between two or more players to avoid the truthful disclosure of $\rev_\P$ and the full payment of $\fee(\rev_\P)$.
    This can only be achieved if all players $\{\Q_i | A_{\Q_i} \in H \land \Q_i \neq \P\}$ agree, in exchange for bribes $\bribe_i$ from $\P$, to not invoke the takeback under \Cref{rul:takeback}.
    For practical reasons we assume that the players $\Q_i$ each agree to lock up collateral $\coll_i$ for time $T$, which is transferred to $\P$ if $\Q_i$ breaks the contract---other colluding players have their collateral returned.
    The collateral is locked up at the same time as receiving the bribe so collusion can be equivalently characterised as $\Q_i$ investing $\coll_i - \bribe_i$ for an investment return of $\bribe_i$.
\end{dfn}

\begin{rem}
    \label{rem:collude-keep}
    A token owner only colludes to avoid the fee if they wish to keep the token, implying that $\val_\P \ge \fmv_\P$.
\end{rem}

\newcommand{\bribeSum}[1]{\sum_{#1=1}^N \bribe_{#1}}

\begin{lem}
    A collusion between rational players is only feasible if the following conditions are met:
    \begin{align}
        \fee(\rev_\P) & > \sum_{i=1}^N \bribe_i \label{eq: total bribe}                      \\
        \bribe_i      & > \roi \coll_i \label{eq: lower bribe}                               \\
        \coll_i       & > \val_\P + \bribeSum{j} \quad\forall i \label{eq: lower collateral}
        ,
    \end{align}
    where $\roi := 1 - e^{-RT}$ and $R > 0$ is the continuously compounding rate of return that can be earned by using the collateral for other purposes.
    $R$ is assumed to be equal for all players and independent of the amount of funds invested.
\end{lem}

\begin{proof}
    \Cref{eq: total bribe} is a direct consequence of the rationality of the current owner of the token.
    If the total bribe would exceed the fee, the owner $\P$ would simply reveal $\rev_\P$ and pay the fee.

    The previous owners of the token $\Q_i$ would only accept a collusion if their bribes $\beta_i$ exceed any potential returns from investing the difference in the locked up collateral, which can be expressed as
    \begin{align}
        \bribe_i        & > (\coll_i - \bribe_i) \left( e^{RT} - 1\right)     \\
        \bribe_i e^{RT} & > \coll_i \left( e^{RT} - 1\right)                  \\
        \bribe_i        & > \coll_i \left( 1 - e^{-RT} \right) = \roi \coll_i
        .
    \end{align}

    \Cref{eq: lower collateral} follows from the desire of the current owner to protect themselves in the case of a breach of contract, where their token gets taken back.
    In this case, the player wants to be compensated for at least the loss of the token, which the owner values at $\val_\P$ (per \Cref{rem:collude-keep}), plus the total bribe spent to incentivise all $\O_i$.
\end{proof}

\begin{thm}[Collusion limit]
    \label{thm: collusion limit}
    The number of colluding players and their collusion time are limited by
    \begin{align}
        T  N < \frac{\fee(\rev_\P)}{R  \val_\P}
    \end{align}
    if $R  T \ll 1$.
\end{thm}

\begin{proof}
    Using \Cref{eq: total bribe,eq: lower bribe,eq: lower collateral} to eliminate $\bribe_i$ and $\coll_i$, we get
    \begin{align}
        \fee(\rev_\P) >
        \bribeSum{i}
        > \roi  \sum_{i=1}^N \coll_i >
        \roi  N \left( \val_\P + \bribeSum{j} \right) >
        \roi  N  \val_\P
        .
    \end{align}
    Since $e^{-RT} = 1 - RT + \mathcal{O}((RT)^2)$, for small $R  T \ll 1$, we can approximate $\roi = 1 - e^{-RT} \approx R  T$ such that
    \begin{align}
        \fee(\rev_\P) > RTN\val_\P, \numberthis
    \end{align}
    and rearranging the terms yields the desired result.
\end{proof}

\begin{rem}
    Henceforth, when referring to a linear fee function we mean $\fee(\rev) = \rho \rev$ with royalty share $\rho \in (0,1)$.
\end{rem}

\begin{cor}
    Using a linear fee function and the identity as price function, the collusion limit in \Cref{thm: collusion limit} becomes independent of token valuation
    \begin{align}
        T  N < \frac{\fee(\rev_\P)}{R  \val_\P} = \frac{\rho  \fmv_\P}{R  \val_\P} \leq \frac{\rho}{R}
        .
    \end{align}
\end{cor}

\begin{rem}
    On Ethereum, the interest rate $R$ is at least the one given by the PoS consensus mechanism.
    For example, using the Lido staking pool, locked up capital could generate a return of $R = \sim 3.5 \%  \text{yr}^{-1}$ \cite{LidoStakingAPR}.
    With a royalty fraction of $\rho = \sim 3.5 \%$, it would thus only be feasible to collude with one person over a year, or with twelve people over a month.
\end{rem}

\newcommand{\feeForPrice}[1]{\fee(\priceInv(#1))}

\begin{lem}[Lower bound on $\fmv_\P$ after a sale]
    \label{lem:fmv-lower-bound}
    If $\Q$ sells a token to $\P$ for a price of $\cost$, the token's free-market value is bounded such that
    \begin{align}
        \label{eq:fmv-lower-bound}
        \fmv_\P \ge \cost - \feeForPrice{\cost}
        .
    \end{align}
\end{lem}

\begin{rem}
    By definition, paying a fee of $\feeForPrice{\rev}$ results in an auto-sale listing price of $\rev$.
\end{rem}

\begin{proof}[Proof of \Cref{lem:fmv-lower-bound}]
    $\Q$ originally owned the token and was willing to sell it in return for $\cost$.
    $\P$ paid $\cost$ to $\Q$ prior to disclosing some $\rev_\P$.
    If $\Q$ then accepts the auto-sale price by paying $\price(\rev_\P)$, and also pays $\feeForPrice{\cost}$ to invalidate a take-back, then $\O$ is back in their original state of token ownership and a potential sale for $\cost$, but with a change in funds of ${\cost - \price(\rev_\P) - \feeForPrice{\cost}}$.
    If positive, this change constitutes an arbitrage in $\O$'s favour.
    As a guaranteed buyer can be found for $\P$'s auto-sale, we have established a lower bound on $\fmv_\P$ at $\O$'s arbitrage threshold of $\cost - \feeForPrice{\cost}$.
\end{proof}

\begin{thm}[Royalty avoidance bound]
    \label{thm:avoidance-bound}
    With the identity as price function, the difference between the fee paid after a sale $\fee(\cost')$ and the idealised sales tax $\fee(\cost)$ is bounded by a factor of $L_\fee$. For a linear fee function, fee avoidance is capped at $\rho^2 \cost$.
\end{thm}

\begin{proof}
    By definition, the best Lipschitz constant of $\fee$, $L_\fee$, limits the absolute difference in fees such that
    \begin{align}
        \abs{\fee(\cost) - \fee(\cost')} \le L_\fee \abs{\cost - \cost'}
        .
    \end{align}
    Substituting the lower bound on $\fmv_\P$ established in \Cref{lem:fmv-lower-bound} for $\cost'$,
    \begin{align}
        \abs{\fee(\cost) - \fee(\cost - \fee(\cost))} & \le L_\fee \abs{\cost - (\cost - \feeForPrice{\cost})} \\
                                                      & = L_\fee \abs{\cost - (\cost - \fee(\cost))}           \\
                                                      & = L_\fee \fee(\cost)
        .
    \end{align}
    For a linear fee function $\fee(\cost) = \rho \cost$ the Lipschitz constant is simply $\rho$, and $L_\fee \fee(\cost) = \rho^2 \cost$.
\end{proof}

\section{Discussion}

\subsection{Implementation, adoption, and user experience}

Although NFT traders are early adopters of novel technology, there is an established status quo that adoption of \erc{7526} would disrupt.
Widespread acceptance will therefore be dependent on careful change management and communication to end users.

\subsubsection{Characterisation as an auction}

Our early discussions with NFT traders revealed some unease with the proposed mechanism due to two main reasons: the perceived additional complexity around self-assessed token valuation and the concept of eventual ownership, stating that it felt as if their assets could be sold against their will.
In attempting to simplify the description we also realised that reframing the mechanism as a form of auction made the lack of immediate ownership more agreeable.

Instead of thinking of NFT sales as atomic and final, \erc{7526} can be understood as a modified English auction.
The seller lists the token for sale at what can be thought of as the reserve, $\cost$.
The buyer accepts the reserve by paying $\cost$, also dictating the next allowable bid according to $\price(\rev)$.
Should another bidder accept the dictated bid, they too lock in a next bid and the process continues.
Unlike a typical auction, the seller only receives the reserve, and differences between intermediate bids are paid to the losing \emph{bidders} as a reward for participation (recall that they also have their $\phi(x_i)$ reimbursed).

Traders who are looking for a quick flip and believe they are privy to ``alpha'' (i.e. an information assymetry) may see this participation reward as an incentive and thus provide market liquidity.

\subsubsection{Marketplace support}

Although it would be in the best interests of buyers, there is no requirement for immediate marketplace integrations.
In the absence of integrations, other avenues for fee payment can be made available, including direct smart-contract interaction.
This, however, would make for a poor user experience through additional steps and a significant probability of accidental non-disclosure resulting in unforeseen take-backs.

We recommend that creators only make use of \erc{7526} once there is at least one supporting marketplace, and that they explicitly communicate the implications of trading elsewhere.
This provides a natural incentive for the first platforms to offer support, and a disincentive for those that don't as use of non-supporting marketplaces will be the root cause of take-backs.

While the mechanism makes it difficult for dishonest users to avoid fee payments, this is a necessary but insufficient condition for successful adoption.
It must also be trivial for rule-abiding users to avoid accidental mistakes that might result in token loss.
We therefore recommend that \erc{7526}-compliant marketplaces forward royalty payments in the same transaction as sales, rendering a take-back impossible.
Marketplaces may also use heuristics to assist buyers in choosing default values for $\rev$, thereby alleviating some of the perceived complexity of the mechanism.

\subsubsection{Choice of $\fee$ and $\pi$}

Our analyses aimed to prove the most general cases possible, placing minimal restrictions on the fee and price functions.
In practice, however, we expect to see linear fee functions $\fee(\rev) = \rho \rev$ with $\rho \in (0,0.1)$ in keeping with current market trends.
We further expect $\price$ to be the identity function as this is the simplest to understand.

\subsubsection{Arbitrage bots for enforcement}

The lower bound on $\fmv_\P$ established in \Cref{lem:fmv-lower-bound} requires a means of automation to be realistic.
While we can't expect sellers to be available at short notice to buy tokens back upon underrepporting, the arbitrage can be exercised through automation whereby services monitor for opportunities and execute buy-backs on behalf of users.

Similarly, we envisage automated detection of other arbitrage opportunities.
The open availability of marketplace order books establishes a lower bound on $\fmv_\P$, commonly referred to as the ``floor price''.
Any buyer attempting to underpay fees can't rationally disclose any value that results in an auto-sale price below the floor.

\subsubsection{Wrapping}

Token wrapping is a concept in which a smart contract becomes the direct owner of a token $t$, issuing a ``wrapped'' derivative $w$ to the previous owner $a_t$, which initiated the wrapping.
The derivative is itself an NFT and can be traded without royalties, finally being owned by $a_w$ who can choose to ``unwrap'' it---the smart contract then transfers $t$ from itself to $a_w$ instead of back to $a_t$.

As \Cref{rul:takeback} is permanent, a wrapping contract would need to pay royalties to remove $a_t$ from $H$ so there is no take-back risk faced by $a_w$.
Even though the smart contract's address becomes a member of $H$ when unwrapping, as long as the contract code doesn't implement take-backs, $a_w$ remains protected.

Although this allows all but one royalty payment to be avoided, it also relies on a broadly trusted wrapper contract---buyers of token $w$ would have to be assured that they undoubtedly have access to $t$.
While expert users can confirm such access, typical buyers would need to place significant trust in unknown third parties, facing risks that they can avoid through royalty payment.

\subsection{Further research}

The parameterisation of $\fee$ based on the duration for which auto-sales are possible is a natural extension to the proposed mechanism.
Even if a buyer exists for some $\fmv_\P$, a sale isn't guaranteed because the buyer might not have enough time to discover the offer during the auto-sale window.
A shorter window therefore acts in favour of players who wish to underpay while keeping the token as the disincentive against underreporting is reduced.
Conversely, the mechanism benefits from extending the period and might therefore accept a lower fee.
We suspect that players may be able balance the incentives by paying an additional fee in return for curtailing auto-sale windows, or vice versa.

\ifACM
\else
    \newenvironment{acks}{\subsection*{Acknowledgements}}{}
\fi

\begin{acks}
    Both authors thank Scott Kominers, Tim Roughgarden and Pranav Garimidi for feedback on this work.
    Arran Schlosberg thanks Ron Lavi for guidance into the world of mechanism design.
    The authors began work on this manuscript while engaged by Proof Holdings.
\end{acks}

\bibliographystyle{ec24style/ACM-Reference-Format}
\bibliography{references}